\documentclass[11pt,a4paper]{amsart}
\usepackage[foot]{amsaddr}
\usepackage{ifxetex}
\ifxetex
  \usepackage[no-math]{fontspec}
\else
\fi
\usepackage{amsmath}
\usepackage{amsfonts}
\usepackage{amssymb}
\usepackage{amsthm}
\usepackage{fullpage}
\usepackage{microtype}

\usepackage[libertine]{newtxmath}
\usepackage[tt=false]{libertine} %
\usepackage{bm}
\usepackage{bbm}
\usepackage[numbers]{natbib}
\usepackage{xcolor}
\usepackage{enumerate} 
\usepackage{enumitem}
\usepackage{tabularx}
\usepackage{array}

\usepackage{times}
\usepackage{soul}
\usepackage{url}
\usepackage{xcolor}
\usepackage[utf8]{inputenc}
\usepackage[small]{caption}
\usepackage{amsmath}
\usepackage{amsthm}
\usepackage{booktabs}
\usepackage{algorithm}
\usepackage[switch]{lineno}
\usepackage[noend]{algpseudocode}
\usepackage{color}
\usepackage{amssymb}
\usepackage{subcaption}
\usepackage{booktabs}
\usepackage{cleveref}
\usepackage[textwidth=4em,textsize=tiny]{todonotes}
\usepackage{xifthen}
\usepackage[T1]{fontenc}
\usepackage{tgtermes}
\usepackage{thmtools}
\usepackage{thm-restate}

\newbool{doubleblind}
\setbool{doubleblind}{false}

\newtheorem{theorem}{Theorem}
\newtheorem{lemma}{Lemma}
\newtheorem{obs}{Observation}
\newtheorem{cl}{Claim}
\newtheorem{fact}{Fact}

\newcommand{\satisfying}[1]{\ensuremath{\mathsf{Sol}({#1})}}
\newcommand{\F}{\ensuremath{\mathsf{F}}}
\newcommand{\numVars}{\ensuremath{\mathsf{n}}}
\newcommand{\numCubes}{\ensuremath{\mathsf{m}}}
\newcommand{\cube}[1]{\F^{#1}}
\newcommand{\width}[1]{\ensuremath{\mathsf{width}(\cube{#1})}}

\newcommand{\toolname}{\ensuremath{\mathsf{pepin}}}
\newcommand{\oldtoolname}{\ensuremath{\mathsf{pepinBinomial}}}
\newcommand{\ganak}{\ensuremath{\mathsf{GANAK}}}
\newcommand{\klm}{\ensuremath{\mathsf{KLM Counter}}}
\newcommand{\dklr}{\ensuremath{\mathsf{DKLR Counter}}}
\newcommand{\dnfapproxmc}{\ensuremath{\mathsf{DNFApproxMC}}}
\newcommand{\eps}{\varepsilon}
\newcommand{\ConstructLazySample}{\ensuremath{\mathsf{ConstructLazySample}}}
\newcommand{\Threshold}{\mathsf{Thresh}}

\newcommand{\Bin}{\mathsf{Binomial}}
\newcommand{\pois}{\mathsf{Poisson}}
\newcommand{\GenerateSamples}{\ensuremath{\mathsf{GenerateSamples}}}
\newcommand{\rpo}{\ensuremath{\mathsf{Process1}}}
\newcommand{\rpt}{\ensuremath{\mathsf{Process2}}}
\newcommand{\rpf}{\ensuremath{\mathsf{Process3}}}
\algnewcommand{\IfThenElse}[3]{%
  \State \algorithmicif\ #1\ \algorithmicthen\ #2\ \algorithmicelse\ #3}

\newcommand{\MVC}{\ensuremath{\mathsf{MVC}}}

\begin{document}

\title{Engineering an Efficient Approximate DNF-Counter}

\ifdoubleblind
  \author{Author(s)}
\else
\author{Mate Soos, Uddalok Sarkar, Divesh Aggarwal, Sourav Chakraborty, Kuldeep S. Meel, Maciej Obremski}

\address[Mate Soos, Kuldeep S. Meel]{University of Toronto, Toronto, Canada. \textnormal{E-mail: \url{soos.mate@gmail.com}, \url{meel@cs.toronto.edu }}}
\address[Divesh Aggarwal, Maciej Obremski]{National University of Singapore, Singapore. \textnormal{E-mail: \url{divesh@comp.nus.edu.sg}, \url{obremski.math@gmail.com}}}
\address[Uddalok Sarkar, Sourav Chakraborty]{Indian Statistical Institute, Kolkata, India. \textnormal{E-mail: \url{chakraborty.sourav@gmail.com }}}
\fi

\begin{abstract}
Model counting is a fundamental problem in many practical applications, including query evaluation in probabilistic databases and failure-probability estimation of
networks. In this work, we focus on a variant of this problem where the underlying 
formula is expressed in the Disjunctive Normal Form (DNF), also known as \#DNF. This problem has been shown to be \#P-complete, making it often intractable to solve exactly. Much research has therefore focused on obtaining approximate solutions, particularly in the form of $(\varepsilon, \delta)$ approximations.

The primary contribution of this paper is a new approach, called {\toolname}, an approximate \#DNF counter that significantly outperforms prior state-of-the-art approaches. Our work is based on the recent breakthrough in the context of the union of sets in the streaming model. We demonstrate the effectiveness of our approach through extensive experiments and show that it provides an affirmative answer to the challenge of efficiently computing \#DNF.
\end{abstract}

\maketitle

\section{Introduction}\label{sec:introduction}
The problem of model counting is fundamental in computer science, where one seeks to compute the total number of solutions to a given set of constraints. In this work, we focus on a variant of this problem where the underlying formula is expressed in Disjunctive Normal Form (DNF), also known as \#DNF. This problem has many practical applications, including query evaluation in probabilistic databases \cite{DalviS07} and failure-probability estimation of networks \cite{Karger2001}. 
The problem of \#DNF is known to be \#P-complete \cite{Valiant79}, where \#P is the class of counting problems for decision problems in NP. Due to the intractability of exact \#DNF, much of the research has been focused on obtaining approximate solutions, particularly in the form of $(\epsilon, \delta)$ approximations, where the count returned by the approximation scheme is within $(1 + \epsilon)$ of the exact count with confidence at least $1 - \delta$.

There has been a significant amount of research on the problem of
approximate \#DNF counting. Karp and Luby~\cite{KarpLuby1983}
proposed the first Fully Polynomial Randomized Approximation Scheme (FPRAS) for \#DNF, known as the {\ensuremath{\mathsf{KL\ Counter}}}. This
was followed by the {\klm} proposed by Karp, Luby, and
Madras~\cite{KLM89} and the 
{\ensuremath{\mathsf{Vazirani\ Counter}}}
proposed by Vazirani \cite{booksVV}.
More recently, Chakraborty, Meel, and Vardi~\cite{CMV16} showed that the hashing-based framework for approximate CNF counting can be
applied to \#DNF, leading to the {\dnfapproxmc} algorithm. This was
subsequently improved upon by Meel, Shrotri, and Vardi~\cite{MSV17,MSV19}
with the design of {\ensuremath{\mathsf{SymbolicDNFApproxMC}}} algorithm.

Given the plethora of approaches with similar complexity, it is natural to wonder how they compare to each other. Meel, Shrotri, and Vardi~\cite{MSV17,MSV19} conducted an extensive study to answer this question, producing a nuanced picture of the performance of these approaches. They observed that there is no single best algorithm that outperforms all others for all classes of formulas and input parameters. 
These results demonstrate a gap between runtime 
performance and theoretical bounds on the time complexity of techniques for approximate \#DNF, thereby highlighting the room for improvement in the design of FPRAS for \#DNF. In particular, they 
left open the challenge of designing an FPRAS that outperforms every other FPRAS.

The primary contribution of this paper is an affirmative answer to the above challenge. We present a new efficient approximate \#DNF counter, called {\toolname}, with (nearly) optimal time complexity that outperforms all the existing FPRAS algorithms when run on standard benchmark data. Our investigations are motivated by the recent breakthrough by Meel, Vinodchandran, and Chakraborty ~\cite{MVC21} on approximating the volume of the union of sets in the streaming model. However, we found their algorithm to be highly impractical due to its reliance on sampling from the Binomial distribution and runtime overhead arising from requirement of a large amount of randomness.

To overcome these barriers, we first demonstrate that sampling from the Poisson distribution suffices to provide theoretical guarantees. We then propose algorithmic engineering innovations, such as a novel sampling scheme and the use of lazy data sampling to improve runtime performance. These innovations allow us to design {\toolname}, a practically efficient approximate \#DNF counter that outperforms every other FPRAS. In particular, over a benchmark suite of 900 instances, {\toolname} attains a PAR-2 score\footnote{PAR-2 score is a penalized average runtime. It assigns a runtime of two times the time limit for each benchmark the tool timed out on.} of 3.9 seconds while all prior techniques have PAR-2 score of over 150, thereby attaining a 40$\times$ speedup. 

The rest of the paper is organized as follows: we present notations and preliminaries in Section~\ref{sec:prelims}. We then present, in Section~\ref{sec:background}, a detailed overview of the prior approaches in the context of streaming that serve as the inspiration for our approach. We present the primary technical contribution, {\toolname}, in Section~\ref{sec:techical} and present detailed empirical analysis in Section~\ref{sec:experiments}. Finally, we conclude in Section~\ref{sec:conclusion}.
\section{Notations and Preliminaries}\label{sec:prelims}

In this paper, we consider Disjunctive Normal Form (DNF) formulas, which are disjunctions over conjunctions of literals. A literal is a variable or the negation of a variable. The disjuncts in a formula are referred to as {\em cubes}, and we use $\cube{i}$ to denote the $i^{th}$ cube. A formula $\F$ with $\numCubes$ cubes can be represented as $\F = \cube{1} \lor \cube{2} \lor ... \lor \cube{\numCubes}$. We use $\numVars$ to denote the number of variables in the formula. The width of a cube $\cube{i}$ is the number of literals it contains and is denoted by $\width{i}$. 

Throughout this paper, $\log$ means logarithms to the base $2$, and $\ln$ means logarithms to the base $e$.
We use $\Pr[A]$ to denote the probability of an event $A$, and $\mu[Y]$ and $\sigma^2[Y]$ to denote the expectation and variance of a random variable $Y$, respectively. An assignment of truth values to the variables in a formula $\F$ is called a satisfying assignment or witness if it makes $\F$ evaluate to true. The set of all satisfying assignments of $\F$ is denoted by $\satisfying{\F}$. Computing a satisfying assignment, if one exists, can be done in polynomial time for DNF formulas. The constrained counting problem is to compute $|\satisfying{\F}|$.

We say that a randomized algorithm $\mathcal{A}$ is an FPRAS for a problem if, given a formula $\F$, a tolerance parameter $\varepsilon \in (0,1)$, and a confidence parameter $\delta \in (0,1)$, $\mathcal{A}$ outputs a random variable $Y$ such that the probability that $\Pr[\frac{1}{1 + \varepsilon}|\satisfying{\F}| \leq Y \leq (1 + \varepsilon)|\satisfying{\F}|] \geq 1-\delta$, and the running time of the algorithm is polynomial in $|\F|$, $1/\varepsilon$, and $\log(1/\delta)$.

\begin{algorithm}[h]
    \caption{{\rpo}}\label{algo:rp1}
    \begin{algorithmic}[1]
    \State $N \leftarrow \pois(|S| p / 2)$
    \State Let $T$ be a multi-set obtained by drawing $N$ samples from the set $S$ with replacement.
    \end{algorithmic}
\end{algorithm}
        
\begin{algorithm}[h]
    \caption{{\rpt}}\label{algo:rp2}
    \begin{algorithmic}[1]
        \State $N \leftarrow \pois(|S| p)$
        \State Let $T$ be a multi-set obtained by drawing $N$ samples from $S$ with replacement
        \State Throw each of the elements in $T$ independently with probability $1 / 2$
    \end{algorithmic}
\end{algorithm}

\begin{algorithm}[h]
    \caption{{\rpf}}\label{algo:rp4}
    \begin{algorithmic}
    \State Initialize $T \leftarrow \varnothing$
    \For{$i=1$ to $n$}
    \State $N \leftarrow \pois(p / 2)$
        \If{$N \geq 1$}
            $\operatorname{add} N$ copies of $s_i$ to $T$
        \EndIf
    \EndFor
    \end{algorithmic}
\end{algorithm}

\begin{cl}
     Let $X$, $Y$ are two independent Poisson random variables such that $X \gets \pois(\lambda_1)$ and $Y \gets \pois(\lambda_2)$ respectively. Then $(X + Y) \gets \pois(\lambda_1 + \lambda_2)$.
     \label{cl:poissonsum}
\end{cl}

\begin{cl}
    For any $S = \{s_1, s_2, \ldots , s_n\}$ and $p$ the distributions of the outputs of the three random processes {\rpo}, {\rpt}, and {\rpf} are equivalent. That is the distribution of their outputs are same.
    \label{cl:poissonrp}
\end{cl}

\begin{proof}
    We will prove the claim by demonstrating the following steps in succession: (1) \(\rpo\) and \(\rpt\) are equivalent, and (2) \(\rpo\) and \(\rpf\) are equivalent.

    \begin{enumerate}
        \item For simplicity, let us assume \(|S|p = \lambda\). Therefore, in \(\rpo\), \(|T|\) is distributed according to \(\pois(\lambda/2)\). After Line 2 of \(\rpt\), \(|T|\) is distributed according to \(\pois(\lambda)\). Since each element is then selected with probability \(1/2\), the probability that \(|T| = k\) for some positive integer \(k\) is given by:
        \begin{align*}
            \Pr\left[|T| = k\right] &= \sum_{t = k}^{\infty} \frac{e^{-\lambda}\lambda^t}{t!} \binom{t}{k}\left(\frac{1}{2}\right)^t \\
            &= \sum_{t = k}^{\infty} \frac{e^{-\lambda}\lambda^t}{t!} \frac{t!}{k!(t-k)!}\left(\frac{1}{2}\right)^t \\
            &= \frac{e^{-\lambda}}{k!}\left(\frac{\lambda}{2}\right)^k \sum_{t = k}^{\infty} \frac{1}{(t-k)!}\left(\frac{\lambda}{2}\right)^{t-k} \\
            &= \frac{e^{-\lambda}}{k!}\left(\frac{\lambda}{2}\right)^k \sum_{j = 0}^{\infty} \frac{1}{j!}\left(\frac{\lambda}{2}\right)^j \\
            &= \frac{e^{-\lambda}}{k!}\left(\frac{\lambda}{2}\right)^k \cdot e^{\lambda/2} \\
            &= \frac{e^{-\lambda/2}}{k!}\left(\frac{\lambda}{2}\right)^k.
        \end{align*}
        Therefore, in \(\rpt\), the size of \(T\) follows the distribution \(\pois(\lambda/2)\). This shows that \(\rpo\) and \(\rpt\) are equivalent.

        \item The equivalence of \(\rpo\) and \(\rpf\) can be proven by applying mathematical induction on \cref{cl:poissonsum}, thus establishing the overall equivalence.
    \end{enumerate}
\end{proof}

\subsection{Related Work}\label{sec:relatedwork}
The problem of designing efficient techniques for \#DNF has a long history. Starting with the work of Stockmeyer~\cite{S83} and Sipser~\cite{Sip83}, randomized polynomial time algorithms for approximately counting various \#P problems were designed. 
In a breakthrough work, Karp and Luby~\cite{KarpLuby1983} introduced the concept of Monte Carlo algorithms for \#DNF. 
Since then, several FPRAS-s based on similar approaches have been developed~\cite{KLM89}. 
All these techniques use a mixture of sampling and carefully updating a counter. 
In recent years, hashing-based techniques have also been used to design FPRASs for \#DNF~\cite{CMV16,MSV17,MSV19}.  While the theoretical results are one part of the story, the practical usability of these FPRAS algorithms gives a different viewpoint.  Various FPRAS algorithms for \#DNF have been implemented, and their performances analyzed and compared in \cite{MSV17,MSV19}. 
There, the authors observed that no single FPRAS algorithm performed significantly better than the rest on all benchmarks.

Applications of \#DNF to probabilistic databases also motivated a number of  algorithms designed for approximate \#DNF that try to optimize
query evaluation \cite{OHK10,FO11,GS14,TSVO04}. These algorithm are, however, either impractical (in terms of time complexity) or are designed to work on restricted classes of formulas such as read-once, monotone, etc.
In addition to the randomized algorithms, a significant amount of effort has gone into designing deterministic approximation algorithms for \#DNF~\cite{LubyV96,Trevisan04,GopalanMR13}.
However, the challenge of developing a fully polynomial time deterministic approximation algorithm for \#DNF remains open~\cite{GopalanMR13}.

\begin{algorithm}[tb]
\caption{\MVC}\label{algo:basic}
\begin{algorithmic}[1]
\State  $\Threshold \gets \left(\frac{\log(12/\delta)+ \log \numCubes}{\varepsilon^2}\right)$
\State $p \gets 1$ ; $\mathcal{X} \gets \emptyset$
\For{$i = 1$ to $\numCubes$}\label{bline:basic-for-sketch-begin}
\For{all $s \in \mathcal{X}$}\label{bline:basic-remove-begin}
\If{$s \models \cube{i}$}\label{bline:basic-check-remove}
{remove $s$ from $\mathcal{X}$\label{bline:basic-remove}}
\EndIf 
\EndFor \label{bline:basic-remove-end}
\State  $N_i \gets \Bin(2^{n-\width{i}}, p)$ \label{bline:cubesize}
\While{$N_i$ + $|\mathcal{X}|$ is more than $\Threshold$}\label{bline:final-check-threshold}
\State Remove each element of $\mathcal{X}$ with probability $1/2$ \label{bline:final-throw}
\State $N_i = \Bin(N_i, 1/2)$ and $p = p/2$ \label{bline:final-update-prob}
\EndWhile
\State  $k=0$ \label{bline:start-sample}
\For{$j = 1$ to $N_i\log N_i$}
\State $s \gets \mathsf{Sample}(\cube{i})$  \label{bline:basic-sample} 
\If{$s \not\in \mathcal{X}$}
\State $\mathcal{X}$.Append($s$). 
\State $k=k+1$
\EndIf 
\If{ $k == N_i$} \textbf{break} \label{bline:end-sample}
\EndIf
\EndFor 
\EndFor \label{bline:final-for-sketch-end}
\State Output $\frac{|\mathcal{X}|}{p}$ \label{bline:output}
\end{algorithmic}
\end{algorithm}

\section{Background}\label{sec:background}
As mentioned in Section~\ref{sec:introduction}, our algorithmic contributions are based on the recent advances in the streaming literature due to Meel, Vinodchandran, and Chakraborty \cite{MVC21}. To put our contributions in context, we review their algorithm, henceforth referred to as {\MVC} after the initials of the authors.%

{\MVC} is a sampling-based algorithm that makes a single pass over the given DNF formula. The high-level idea of the algorithm is to maintain a tuple $(\mathcal{X},p)$ wherein $\mathcal{X}$ is a set of satisfying assignments while $p$ indicates the probability with which every satisfying assignment of $\F$ is in $\mathcal{X}$. Since the number of solutions of $\F$ is not known a priori,  the value of  probability $p$ is not a predetermined value but changes as the algorithm
proceeds. 

We now provide a description of {\MVC}, whose pseudocode is presented in Algorithm~\ref{algo:basic}. {\MVC} processes each cube sequentially and for every cube, it first removes all the solutions of $\F^i$ that belong to $\mathcal{X}$ (lines~\ref{bline:basic-remove-begin} -- \ref{bline:basic-remove-end}). In Line~\ref{bline:cubesize}, we determine the number of solutions $N_i$ that would be sampled from $\F^i$ if each solution  of $\F^i$ was sampled (independently) with probability $p$. The distribution over the number $N_i$ is simulated by the Binomial distribution. Since we do not want to store more than $\Threshold$ elements in $\mathcal{X}$, if $|\mathcal{X}|+N_i$ is larger than $\Threshold$ we decrease $p$ and appropriately adjust $N_i$ (by sampling from Binomial($N_i,p$) and $\mathcal{X}$ (by removing each element of $\mathcal{X}$ with probability $1/2$). This is done in Lines~\ref{bline:final-check-threshold} to \ref{bline:final-update-prob}. We now need to sample $N_i$ distinct solutions of $\cube{i}$ uniformly at random: to accomplish this task, in Lines~\ref{bline:start-sample} -- \ref{bline:end-sample}, we simply pick solutions of $\cube{i}$ uniformly at random with replacement until we have either generated $N_i$ distinct solutions or the number of samples (with replacement) exceeds $N_i \log N_i$. Finally, in Line~\ref{bline:output}, we return our estimate as the ratio of $\frac{|\mathcal{X}|}{p}$. We refer the reader to~\cite{MVC21} for the theoretical analysis of {\MVC}. It is worth remarking the worst-case time complexity of {\MVC} is $\left(2\numVars\cdot (\log (12\numCubes/\delta))^2\log\log (12\numCubes/\delta) \cdot\varepsilon^{-2}\log \varepsilon^{-1}\right)$.

Upon observing the existence of a new algorithm, our first step was to determine whether such an algorithm can translate to practical techniques for DNF counting. However, rather surprisingly, the resulting implementation could only handle a few hundred variables. The primary bottleneck to scalability is the reliance of {\MVC}'s algorithm on the subroutine $\Bin(k,p)$ in line~\ref{bline:cubesize}. State-of-the-art arbitrary precision libraries take prohibitively long time sampling from $\Bin$ when the first argument is of the order $2^{100}$, which is unfortunately necessary to handle formulas with more than a hundred variables. To further emphasize the overhead due to sampling from $\Bin$, a run of the algorithm would invoke $\Bin$ roughly $\numCubes$ times, and every such invocation when the first argument is of the order of $2^{100}$ is prohibitively slow to handle instances in practice. 
Since $\numCubes$ for DNF instances is in the range of few ten to hundred thousand, such a scheme is impractical in contrast with state-of-the-art techniques that could handle such formulas in the order of a few seconds. At this point, it is worth remarking that the crucial underlying idea of the algorithm is to be able to sample every satisfying assignment of $\cube{i}$ with probability $p$, and the current analysis of {\MVC} crucially relies on the usage of the Binomial distribution.      
Consequently, this raises the questions: {\em Is it possible to design an efficient algorithmic scheme based on the underlying ideas that can also lend itself to practical implementation?} 

\section{Technical Contributions}\label{sec:techical}

The primary contribution of our work is to resolve the aforementioned challenge. To this end, we present a new algorithmic scheme, {\toolname}, that achieves significant runtime improvements over state-of-the-art techniques. As a first step, we seek to address the major bottleneck of {\MVC}: avoiding dependence on $\Bin$ by proposing a different sampling routine which no longer ensures that every solution of a given cube $\cube{i}$ is sampled independently with probability $p$.
In order to achieve a significant runtime performance improvement, we profiled our implementation and discovered that sampling from every cube was the most expensive operation. As a remedy, we propose, inspired by lazy (vs eager) lemma proof generation in modern SMT solvers, {\em lazy sampling}  to delay sampling as much as possible without losing correctness (Section~\ref{sec:lazy}). We then discuss several low-level but crucial enhancements in the implementation of {\toolname}. Finally, we close the section with a theoretical analysis of the correctness of {\toolname}.

\begin{algorithm}[tb]
\caption{$\toolname(F,\varepsilon,\delta)$}\label{algo:afinal}
\label{alg:binomial-compute}
\begin{algorithmic}[1]

\State  $\Threshold \gets \max\left(  12  \cdot \frac{\ln (24/\delta)}{\varepsilon^2}, 6 (\ln \frac{6}{\delta} + \ln m) \right)$
\State  $p \gets 1$ ;   $\mathcal{X} \gets \emptyset$

    \For{$i = 1$ to $\numCubes$}\label{line:final-for-sketch-begin}
        \State $t \gets 2^{n-\width{i}}$
        \For{$s\in \mathcal{X}$}
            \If{$s \models \cube{i}$}
            remove $s$ from $\mathcal{X}$\label{bline:final-remove}
            \EndIf 
        \EndFor
        
        \While{$p \geq \frac{\Threshold}{t}$}
            \State Remove every element of $\mathcal{X}$ with prob. $1/2$ 
            \State $p = p/2$
        \EndWhile
        \State $N_i \gets \pois(t\cdot p)$
        \While{$N_i$ + $|\mathcal{X}| > \Threshold$}\label{line:final-check-threshold}
            \State Remove every element of $\mathcal{X}$ with prob. $1/2$ \label{line:final-throw}
            \State $N_i =  \pois(t\cdot p/2)$ and $p = p/2$ \label{line:final-update-prob}
        \EndWhile 
        \State $S \gets \GenerateSamples (\F^i, N_i)$
        \State $\mathcal{X}.\mathsf{Append}$(S)
    \EndFor \label{line:final-for-sketch-end}
    \State Output $|\mathcal{X}|/p$

\end{algorithmic}
\end{algorithm}

\begin{algorithm}
\caption{$\mathsf{GenerateSamples}(\F^i, N_i)$}\label{alg:generatesamples}
    \begin{algorithmic}[1]
    \State $S \gets \emptyset$
    \For{$j = 1$ to $N_i$}\label{bline:generate-high-begin}
        \State $s \gets \ConstructLazySample(\cube{i})$
        \State $S.\mathsf{Append}(s)$
    \EndFor\label{bline:generate-high-end}
    \State \Return $S$
    \end{algorithmic}
\end{algorithm}

\subsection{Subroutine {\GenerateSamples}}\label{sec:lazy}
As mentioned earlier, the above-proposed sampling scheme allows our algorithm to be on par with the prior state-of-the-art techniques. To achieve further speedup, we observed that the subroutine $\mathsf{Sample}(\cube{i})$ often takes over 99\% of the runtime. Therefore, one wonders whether it is possible to {\em not sample}? At the outset, such a proposal seems counterintuitive as after all, {\toolname} is a sampling-based technique. Upon further investigation, two observations stand out: (1) almost all samples generated by the $\mathsf{Sample}$ routine are removed in line~\ref{bline:final-remove} at some point in the future, and (2) to determine whether to remove $s$ from $\mathcal{X}$, one needs to only determine whether $s \models \cube{i}$, which does not require one to know the assignment to all variables in $s$. Consequently, it is only required to generate assignments to variables in order to check whether $s \models \cube{i}$. We achieve such a design in the subroutine {\GenerateSamples}, which we describe next. 

The subroutine {\GenerateSamples} is presented in Algorithm~\ref{alg:generatesamples}. The primary challenge in {\GenerateSamples} is to handle the generation of $N_i$ solutions (not all distinct) randomly from $\cube{i}$ as if we delay the generation of assignments to unassigned variables in $\cube{i}$, then we would not know whether we have generated $N_i$ solutions. 
{\ConstructLazySample} sets the value to only the variables that appear in $\cube{i}$ and for the rest of the variables, it sets them to a special symbol MARK, that is, $s$ is a mapping from the set of variables to $\{$TRUE, FALSE, and MARK$\}$.
Therefore, in contrast to relying on the expensive operation of pseudorandom generation, we can compute and store $s$  at extremely high speed.  Overall, we have deferred assignment to variables in $s$ (except the ones corresponding to literals in $\cube{i}$) at the time when we are required to check whether $s \models \cube{j}$ when a new cube $\cube{j}$ arrives. At such a time, for all the variables that are set to MARK in $s$ but whose values are fixed in $\cube{j}$, we use the pseudorandom generator to generate a random value for the corresponding variables. Note that once we have assigned all the variables corresponding to literals in $\cube{j}$, we can perform the check whether $s \models \cube{j}$ by only checking whether  $s$ and $\cube{j}$ agree on assignment to variables corresponding to literals in $\cube{j}$. If $s$ and $\cube{j}$ do agree on all such variables, we can remove $s$, which showcases the strength of our approach as we could avoid all the work required to assign the variables in $s$ that are still set to MARK.

\subsection{Engineering Enhancements}

\paragraph{Dense Matrix-based Sample Storage}
Sample storage for all samples is stored in a single contiguous pre-allocated memory array, similar to a dense matrix representation. This helps with cache locality and ensures that when we check the samples to be emptied, we go forward, and only forward, in memory, with fixed jump sizes. This allows the memory subsystem to prefetch values the algorithm will read from memory, thereby masking memory latency, where memory latency can often be over 100x slower than instruction throughput in modern CPUs.

The current maximum number of samples always stays allocated, and we keep a stack where we have the next empty slot. When a sample is removed, we simply put their number on this stack. The size of the stack tells us the number of empty slots (i.e. unfilled sample slots).

Sample storage is further bit-packed. Each variable's value in the sample is represented as 2 bits, as we need to be able to represent not only TRUE and FALSE but also MARK. The bit representation used is 00 = FALSE, 01 = TRUE, 11 = MARK, which allows us to quickly set all-MARK by filling the vector with 1's using highly optimized, SSE memset operations.

\paragraph{Sparse Matrix-based Sample Storage} Since a large portion of the samples contain MARK values, one may ask whether it would be faster to use a sparse matrix representation where only 1's and 0's are stored, along with the number of consecutive MARKs following the 1 or 0. To check whether such a system would be faster, we have also developed an implementation that uses such a sparse matrix representation. Unfortunately, this implementation is very slow for anything but extremely sparse DNFs. We compare its performance to the dense matrix representation in Section \ref{sec:experiments}.

\paragraph{Handling Arbitrary Precision}
We made extensive use of the GNU Bignum library~\cite{GNUBignum} for all values that need high precision. We use MPQ for fractions such as sampling probabilities, and MPZ for large numbers such as the precision product. All bignum variables are pre-allocated and pre-initialized and, when appropriate, re-used to reduce dynamic memory allocation. Furthermore, observe that the sampling probability is always of the form $2^{-k}$ for integer $k$. Therefore, we only keep the exponent bits $k$ and regenerate sample probability when needed. Since the GNU Bignum library has a special function to quickly generate values of the form $2^{-k}$, such re-generation is fast.

\subsection{Theoretical Analysis of the Algorithm}

We will need the following bounds.%

\begin{restatable}[Chernoff Bound]{lemma}{chernoff}
\label{lem:chernoff}
Let $X \gets \pois(\lambda)$, for some $\lambda>0$. Then for any $x > 0$ following two inequalities hold.
\[
\Pr[|X - \lambda| \geq x] \le 2e^{-\frac{x^2}{2\lambda + x}}\;
\]
\[
\Pr[X - \lambda \geq x] \le e^{-\frac{x^2}{2\lambda + x}}\;
\]
\end{restatable}

Proof of \cref{lem:chernoff} is deferred to the appendix.
Now, we complete the proof of our main result.

\begin{theorem}
\label{thm:main}
For any $\eps\leq 1$ and any $\delta \in (0,1)$ the Algorithm~\ref{algo:afinal}  outputs a number in $({(1-\eps)}|\mathsf{Sol(F)}|,(1+\eps)|\mathsf{Sol(F)}|)$ with probability at least $1-\delta$ and runs in expected time
\[
O\left(\frac{1}{\eps^2} m n \log \frac{1}{\delta} \left(\log \frac{1}{\eps} + \log \frac{1}{\delta} + \log m\right)\right) \;.
\]
\end{theorem}

\begin{proof}
    Let $S_i$ be the set of all satisfying assignments of $\mathsf{F}_i := \cube{1} \vee \cdots \vee \cube{i}$. To begin with, we will inductively define the array $X_j^{(i)}$ of distributions that we may get over the sets $S_i$ during the execution of the algorithm for $i \in \{1,\ldots,m\}$, and $j$ being a non-negative integer. The set $X_0^{(1)}$ is defined to be $S_1$. Then we consider two sampling moves:
    \begin{description}
    \item[Decrease $p$ move:] To obtain $X_{j+1}^{(i)}$ from $X_j^{(i)}$, we sample each element of $X_j^{(i)}$ independently with probability $1/2$. This corresponds to Step 8, and Step 12 in Algorithm~\ref{algo:afinal}. 
    \item[Next set move:] To obtain $X_{j}^{(i+1)}$ from $X_j^{(i)}$, we remove all elements $X_{j}^{(i)}$ that is a satisfying assignment of $\cube{i+1}$ (corresponding to step 6 of the Algorithm~\ref{algo:afinal}), and sample (with replacement) $N_i$ satisfying assignments of $\cube{i+1}$ and add those to $X_{j}^{(i)}$ (corresponding to step 10, 13, and 14 of Algorithm~\ref{algo:afinal}), where $N_i$ is drawn from $\pois\left(\frac{|\cube{i+1}|}{2^j}\right)$. Therefore by claim~\ref{cl:poissonrp}, for all elements $s \in \cube{i+1}$, $X_{j}^{(i+1)}$ contains $k_j(s)$ number of copies of $s$, where $k_j(s)$ drawn independently from $\pois\left(\frac{1}{2^j}\right)$.
    \end{description}
    
    We make the following simple observation. 
    \begin{obs}
    For all $i,j$, the random multiset $X_{j}^{(i)}$ contains samples from set $S_i$ and for each satisfying assignment $s \in S_i$, $S_i$ contains $k_j(s)$ many copies of $s$ where $k_j$ is drawn independently from $\pois\left(\frac{1}{2^j}\right)$.
    \end{obs}
    
    For $1 \leq i \le m$, and $j \ge 0$, we define the event $E_j^{(i)}$ as the event that after the algorithm processes the first $i$ cubes, the value of $p$ is $\frac{1}{2^j}$. These events can depend arbitrarily on all the random sets $X_{j}^{(i)}$'s sampled during the execution of the algorithm. This corresponds to an arbitrary sequence of the aforementioned moves over the array of distributions $X_j^{(i)}$ depending on the decisions of the Algorithm~\ref{algo:afinal} in Steps 7, and 11. 
    
    Let us define another event $A_j^{(i)}$ as follows,
    $$A_j^{(i)} := \left| X_{j}^{(i)} \right| \notin \Big[|S_i|2^{-j}{(1-\eps)}, |S_i|2^{-j}(1+\eps)\Big]$$

    Let $j^*$ be the smallest $j$ such that $\frac{1}{2^j} < \frac{\mathsf{Thresh}}{4|S_m|}$. 
    
    We note that the probability that the algorithm outputs a number that is not within the bound $\Big[|S_m|{(1-\eps)}, |S_m|(1+\eps)\Big]$ is given by,
    \begin{align*}
    \Pr[\bigcup_{j=0}^\infty (E_j^{(m)} \wedge A_j^{(m)})] &\le \sum_{j=0}^{j^* - 1} \Pr[ A_j^{(m)}] + \Pr[\bigcup_{j \ge j^*} (E_j^{(m)}] \;.
    \end{align*} 
    
    We bound $\sum_{j=0}^{j^* - 1} \Pr[ A_j^{(m)}]$ and $\Pr[\bigcup_{j \ge j^*} E_j^{(m)}]$.
    
    From the above observation for any fixed $j$, $X^m_j$ contains $k_j$ copies of each satisfying assignment of $S_m$. Therefore from claim~\ref{cl:poissonrp} the size of $X^m_j$ is distributed according to $\pois\left(\frac{|S_m|}{2^j}\right)$.
    
    By Lemma~\ref{lem:chernoff}, we have that since $|S_m|\frac{1}{2^{j^* - 1}} > \frac{\mathsf{Thresh}}{4}$ 
    \begin{align*}
    \sum_{j=0}^{j^* - 1} \Pr[ A_j^{(m)}] &\le 2e^{-\eps^2\mathsf{Thresh}/12} + 2 e^{-\eps^2 \mathsf{Thresh}/6} + \cdots \\
    &< 4e^{-\eps^2 \mathsf{Thresh}/12} \le \frac{\delta}{6} 
    \;,
    \end{align*}
    where we use that $ \mathsf{Thresh} > \frac{12 \ln (24/\delta)}{\eps^2}$ and $\eps \leq 1$. %
    
    To bound $\Pr[\bigcup_{j \ge j^*} E_j^{(m)}]$, we notice that the event $\bigcup_{j \ge j^*} E_j^{(m)}$ occurs only if the decrease $p$ move happens from $X_{(j^* - 1)}^{(i)}$ to $X_{(j^* )}^{(i)}$ for some $i$. This happens if $\pois\left(\frac{|S_i|}{2^{j^*-1}}\right)$ is larger than $\mathsf{Thresh}$. Since $\frac{1}{2^{j^* - 1}} < \frac{\mathsf{Thresh}}{2|S_m|}$, and so $\mathsf{Thresh} - |S_i| \cdot \frac{1}{2^{j^*-1}} > \frac{\mathsf{Thresh}}{2}$, by Lemma~\ref{lem:chernoff}, the probability of this event is at most
    \[
    e^{-\mathsf{Thresh}/6} \le \frac{\delta}{6m} \;,
    \]
    where we use that $\mathsf{Thresh} \ge 6 \left(\ln \frac{6}{\delta} + \ln m\right)$.  By a union bound over $i$, we have that 
    \[
    \Pr[\bigcup_{j \ge j^*} E_j^{(m)}] \le m \cdot \frac{\delta}{6m} = \frac{\delta}{6} \;. 
    \]
    Therefore the desired bound on the error probability follows.
    Now we bound the time complexity. The steps 5, 6 are executed at most $m \cdot|\mathcal X| < m\cdot \Threshold$ times. Since, the probability that $p$ cannot drop below $\frac{1}{2^n}$ in steps 7-9, the steps 7-9 are executed at most $n+m$ times. Since, after step 10, $\mathbb{E}[N_i] < \Threshold$, and $|\mathcal X| \le \Threshold$, steps 11-13 are executed $O(1)$ times for every $i$, and hence $O(m)$ times in total. So, the overall time complexity is dominated by Steps 10 and 14, and hence is bounded by
    \begin{align*}
    &m \cdot \Threshold (\ln \Threshold + \ln \frac{6}{\delta} + \ln m) \cdot n \\
    & \;\;\; = O\left(\frac{1}{\eps^2} m n \log \frac{1}{\delta} \left(\log \frac{1}{\eps} + \log \frac{1}{\delta} + \log m\right)\right) \;.  
    \end{align*}
    
    \end{proof}

\section{Empirical Evaluation}\label{sec:experiments}

We followed the methodology outlined in ~\cite{MSV19}: we use the same benchmark generation tool with the parameters specified by the authors. In particular, the value of $n$ varied from 100 to 100'000 while the value of $m$ varied from $300$ to $8 \times 10^5$ and the width of cubes varied from $3$ to $43$. Furthermore, in line with prior work, we set $\varepsilon$ to 0.8 and $\delta$ to $0.36$. We compare the runtime performance of {\toolname}\footnote{The {\toolname} is available open-source at
\protect\url{https://github.com/meelgroup/pepin}} with the prior state of 
the art techniques, {\klm}~\cite{KLM89}, 
{\dklr}~\cite{DKLR00}, and 
{\dnfapproxmc}~\cite{MSV17}. These techniques were observed to be incomparable to each other while outperforming the rest of the alternatives in ~\cite{MSV19}. All our experiments were conducted on a high-performance computer cluster, each node consisting of 2xE5-2690v3 CPUs with 2x12 real cores and 96GB of RAM, i.e., 4GB limit per run. The timeout was set to be 500 seconds for all runs.

The primary objective of our empirical evaluation was to answer the following questions:
\begin{description}
\item[RQ 1] How does the runtime performance of {\toolname} compare to the state of the art tools {\klm}, {\dklr}, and {\dnfapproxmc}? 
\item[RQ 2] How accurate are the estimates computed by {\toolname}? 
\end{description}

In summary, we observe that {\toolname} achieves significant runtime performance improvements over prior state-of-the-art. In particular, {\toolname} achieves a PAR-2 score of 3.9 seconds while the prior state-of-the-art technique could only achieve a PAR-2 score of 158 seconds, thereby achieving a nearly 40$\times$ speedup. Furthermore, we observe that the observed $\varepsilon$ is only $0.10$ -- significantly lower than $\varepsilon=0.8$. 

\subsection{Performance Experiments}
We follow the methodology of \cite{MSV19} in the presentation and analysis of the results. Accordingly, we first present the cactus plots of performance comparisons in  Fig.~\ref{fig:perf-exp}. The first three subfigures show the performance of the counters on DNFs with different cube widths since the cube width has a significant effect on the performance of all solvers. The final subfigure shows the performance over all the cube widths. 

The graph in Fig. \ref{fig:perf-comp-3} shows that at small cube widths, the previous set of counters all exhibit comparable, and relatively poor, performance, with {\dklr} performing the worst. In fact, even the best of the previous set of counters, {\klm}, only managed to count 60 instances (for cube width 3) within the 500s timeout. In contrast, {\toolname} shows remarkable performance here: it finished for all 180 instances for cube width 3, all under 25 seconds. 
This is primarily due to its lazy sample generation technique, which skips generating random values for variables not in the DNF clause, resulting in many saved computations for cube width 3.

As the cube width increases in Figs. \ref{fig:perf-comp-13} and \ref{fig:perf-comp-43}, the performance of {\toolname} stays similarly good, while the other counters start exhibiting better performance, but never catch up to {\toolname}. For larger cube widths, the lazy sampling starts to be less relevant while the careful design and implementation of the counter plays a more significant role.

Finally, Fig. \ref{fig:perf-comp-all} shows the performance of all the counters over all cube widths. Here, it is clear that {\toolname} outperforms all the counters by a large margin. In fact, {\toolname} is faster than any other counter on all files, except for 27 files, 24 of which are under 1s slower to count by {\toolname}, and the remaining 3 are only under 3s slower to count.

We also used the PAR-2 score to analyze the performance comparison of {\toolname}. PAR-2 score is a penalized average runtime. It assigns a runtime of two times the time limit for each benchmark the tool timed out on.
We present the PAR-2 scores for all the counters in Table~\ref{tbl:PAR2-scores}. As shown in Table~\ref{tbl:PAR2-scores}, {\toolname} outperforms all other solvers by a wide margin, thereby, presenting an affirmative answer to the challenge posed by~\cite{MSV19}.

\begin{table}[tb]
	\centering
	\begin{tabular}{l @{\extracolsep{\fill}} cc}
 \toprule
Counter & finished & PAR-2 score\\
 \midrule
{\toolname}         &900     &2.16\\
{\klm}              &690     &130\\
{\dklr}              &641     &168\\
{\dnfapproxmc}             &270     &375\\
\bottomrule
\end{tabular}
\caption{Comparing the PAR2 scores of the different counters over the performance problems.}
\label{tbl:PAR2-scores}
\end{table}

\begin{figure*}[htb]
    \centering
    \begin{subfigure}[b]{0.4\textwidth}
        \centering
        \includegraphics[scale=0.40]{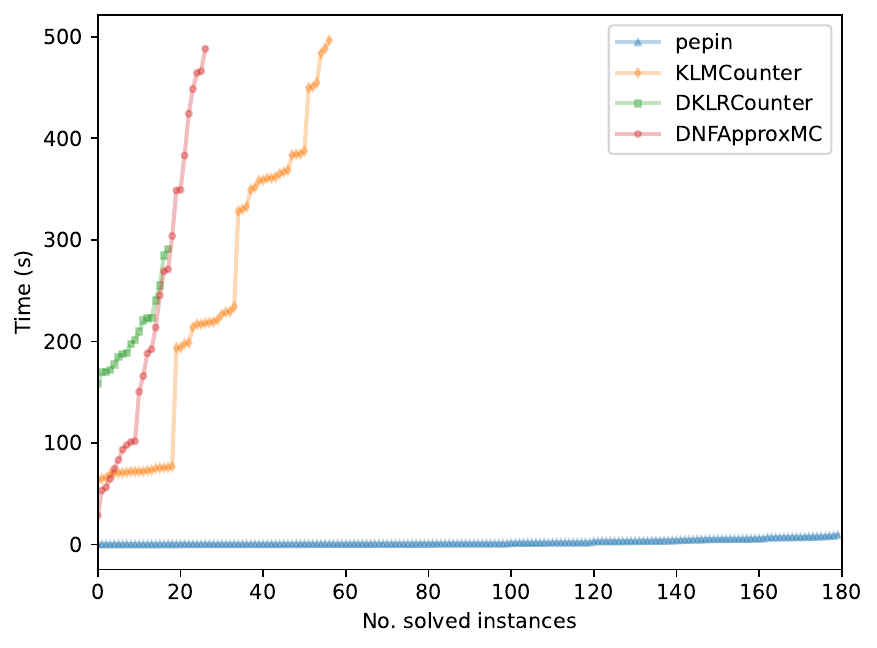}
        \caption[Cubes]%
        {{\small Cube of size 3}}    
        \label{fig:perf-comp-3}
    \end{subfigure}
    \hfill
    \begin{subfigure}[b]{0.4\textwidth}  
        \centering 
        \includegraphics[scale=0.40]{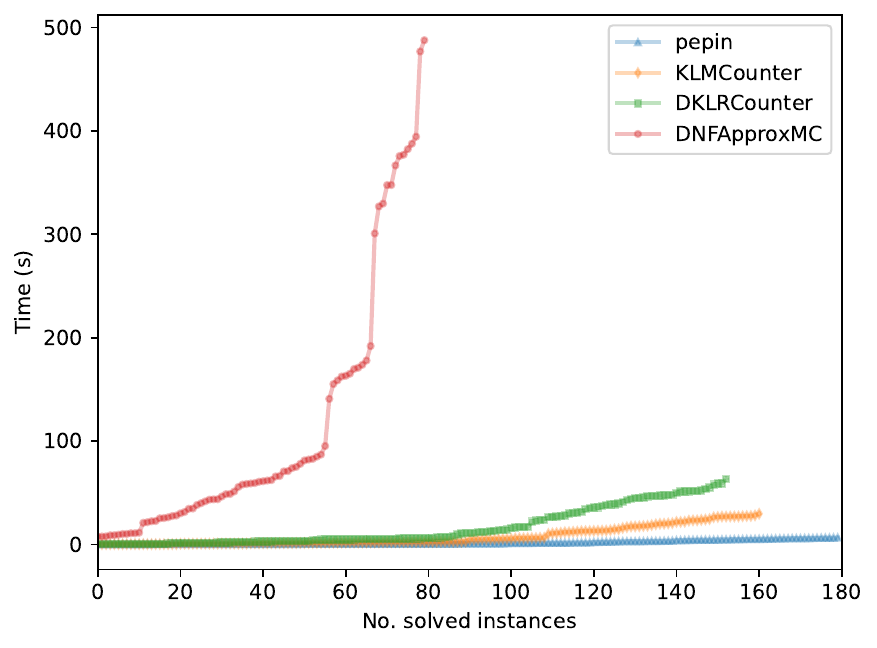}
        \caption[]%
        {{\small Cube of size 13}}    
        \label{fig:perf-comp-13}
    \end{subfigure}
    \vskip\baselineskip
    \begin{subfigure}[b]{0.4\textwidth}   
        \centering 
        \includegraphics[scale=0.40]{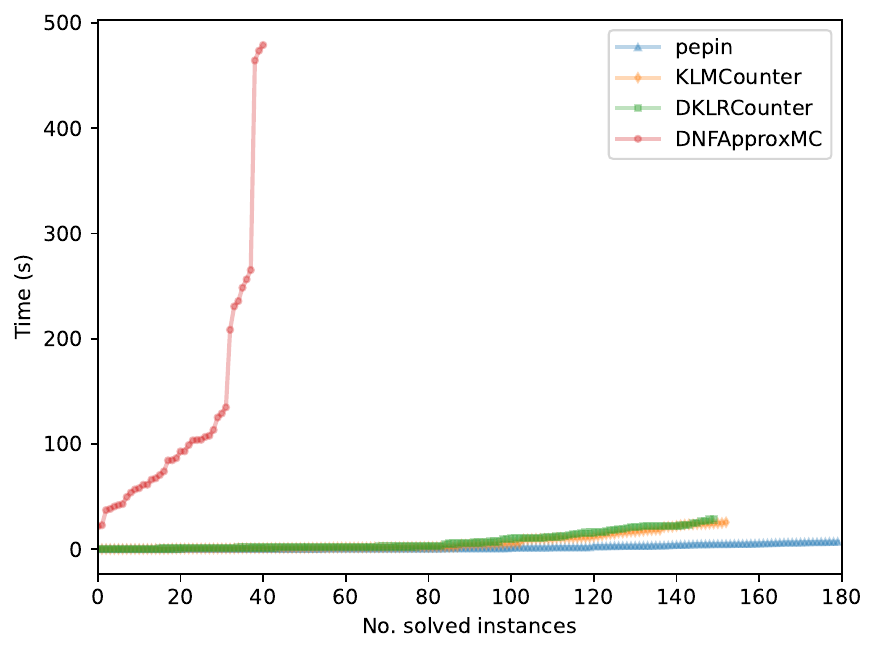}
        \caption[]%
        {{\small Cube of size 43}}    
        \label{fig:perf-comp-43}
    \end{subfigure}
    \hfill
    \begin{subfigure}[b]{0.4\textwidth}   
        \centering 
        \includegraphics[scale=0.40]{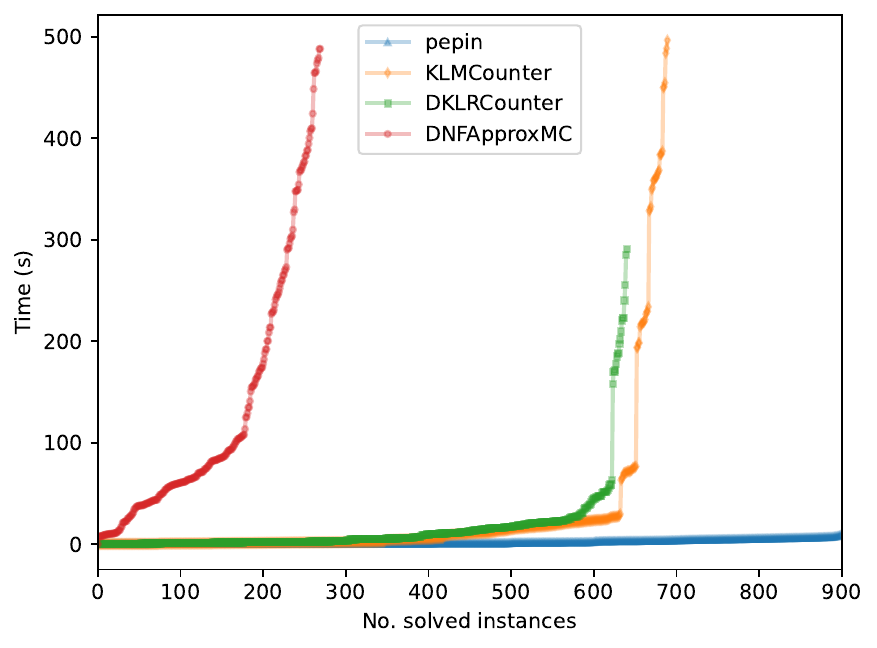}
        \caption[]%
        {{\small All Cubes}}    
        \label{fig:perf-comp-all}
    \end{subfigure}
    \caption{Performance comparison of {\toolname} against the other counters, with different cube widths. As can be seen on the included plots, the cube width matters greatly for most counters other than {\toolname}. This is due to the sparse sampling strategy of {\toolname}.} 
    \label{fig:perf-exp}
\end{figure*}

\subsection{Accuracy Experiments}
To measure the accuracy of {\toolname}, we compared the counts returned by {\toolname} with that of the exact counter, {\ganak}, for all the instances for which {\ganak} could terminate successfully.
Figure~\ref{fig:prec-exp-ganak} shows the relative error while computing the counts by {\toolname}. The $y$-axis represents the relative error of the counts. The tolerance factor ($\varepsilon = 0.8$) is denoted by a red straight line. 
We observed that for all the instances, {\toolname} computed counts within the tolerance. Furthermore, the average mean of observed error for all benchmarks is $0.102$-- significantly better than the theoretical guarantee of $\varepsilon = 0.8$.

\begin{figure}[tb]
\centering
\includegraphics[scale=0.4]{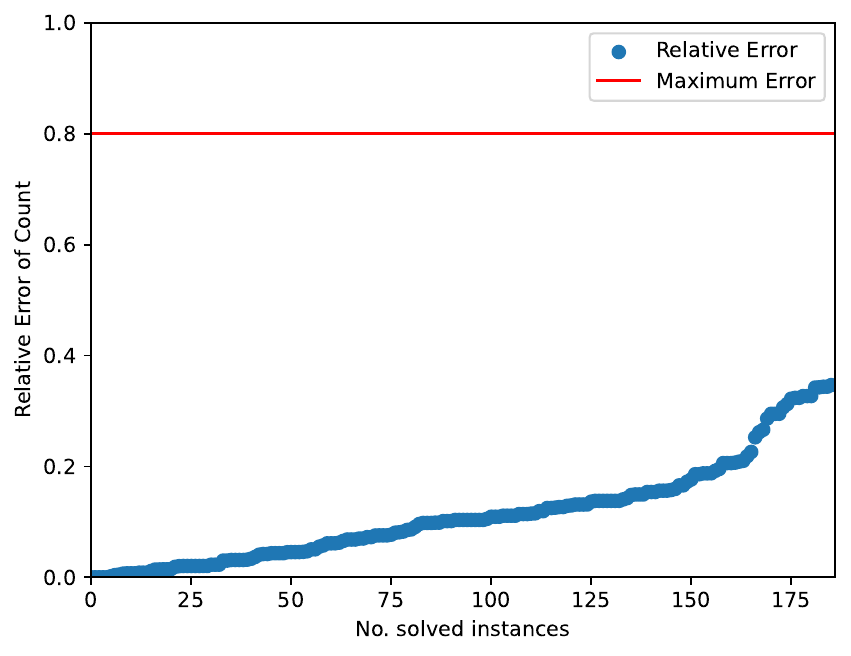}
\caption{The count returned by {\toolname} compared to the exact counter GANAK. All counts of {\toolname} were well within the 80\% permissible error rate as dictated by $\eps=0.8$}
\label{fig:prec-exp-ganak}
\end{figure}

\section{Conclusion}\label{sec:conclusion}
In this paper, we successfully tackled the challenge of designing an FPRAS for \#DNF that outperforms other FPRAS in practice. 
The problem of approximate counting of Disjunctive Normal Form (DNF) formulas, also known as \#DNF counting, has received significant attention from both theoretical and applied researchers. 
However, the challenge of designing an FPRAS that outperforms other FPRASs in practice remained open.
In this paper, we resolve the aforementioned challenge by presenting a new FPRAS called {\toolname} that achieves significant performance improvements on top of the prior state-of-the-art techniques. 
An intriguing direction for future work is to enhance the data structures used in {\toolname}.

\section*{Acknowledgements}
This work was supported in part by the National Research Foundation Singapore under its NRF Fellowship Programme[NRF-NRFFAI1-2019-0004 ], Ministry of Education Singapore Tier 2 grant MOE-T2EP20121-0011,  Ministry of Education Singapore Tier 1 Grant [R-252-000-B59-114 ], and Amazon Research Award.  The computational work for this article was performed on resources of the National Supercomputing Centre, Singapore https://www.nscc.sg
DA was supported by the bridging grant at the Centre for Quantum Technologies titled ``Quantum algorithms, complexity, and communication". MO was supported by the MOE-Tier 1 grant titled ``Non-Malleability and Randomness- Foundations of
Leakage and Tamper Resilient Cryptography" [A-8000614-00-00].

 \clearpage

\appendix

\section{Evaluating {\toolname} Against Previous Versions}

In this section, we present a comprehensive performance comparison between {\toolname} and its predecessor~\cite{SACMO23}. 
The previous version allowed access only to the Binomial distribution, instead of the Poisson distribution, which is costly for arbitrarily large (and small) parameters. We will refer to this version as {\oldtoolname}.
Our evaluation focuses on the runtime efficiency of both tools across a diverse set of 900 instances. 

For the evaluation, we followed a similar setup as those used in the earlier performance experiments with {\toolname}. Specifically, we utilized the same 900 benchmarks that were employed during the performance evaluation of {\toolname}. These benchmarks varied widely in their parameters: the value of \(n\) ranged from 100 to 100,000, the value of \(m\) ranged from 300 to \(8 \times 10^5\), and the width of cubes varied from 3 to 43. Additionally, we set the input parameters \(\delta\) to 0.36 and \(\eps\) to 0.8.

\begin{figure}[h]
	\centering
    \begin{subfigure}[b]{0.4\textwidth}
        \centering
        \includegraphics[scale=0.40]{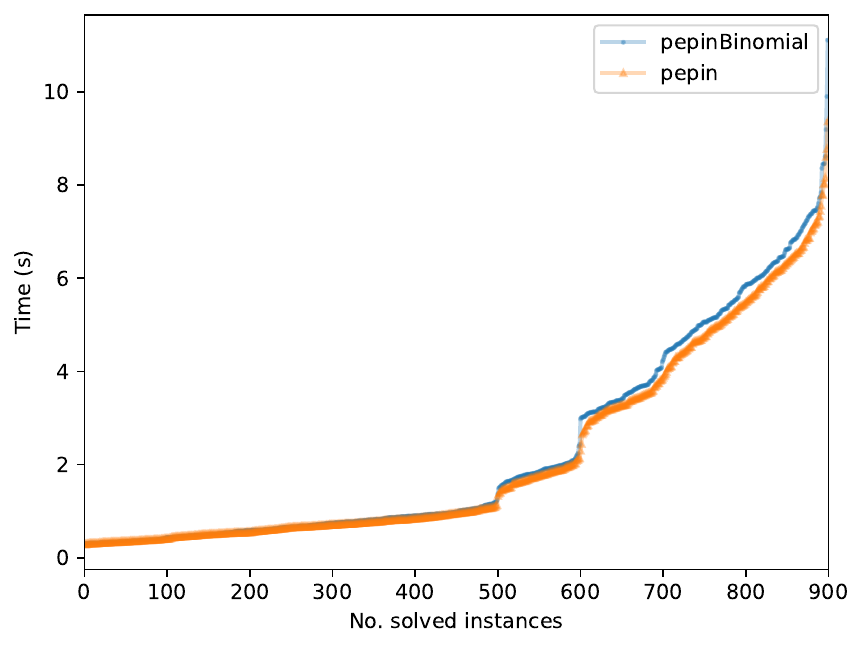}
		\caption{Cumulative time comparison}
        \label{fig:pepin-old-1}
    \end{subfigure}
	\centering
    \begin{subfigure}[b]{0.4\textwidth}
        \centering
        \includegraphics[scale=0.40]{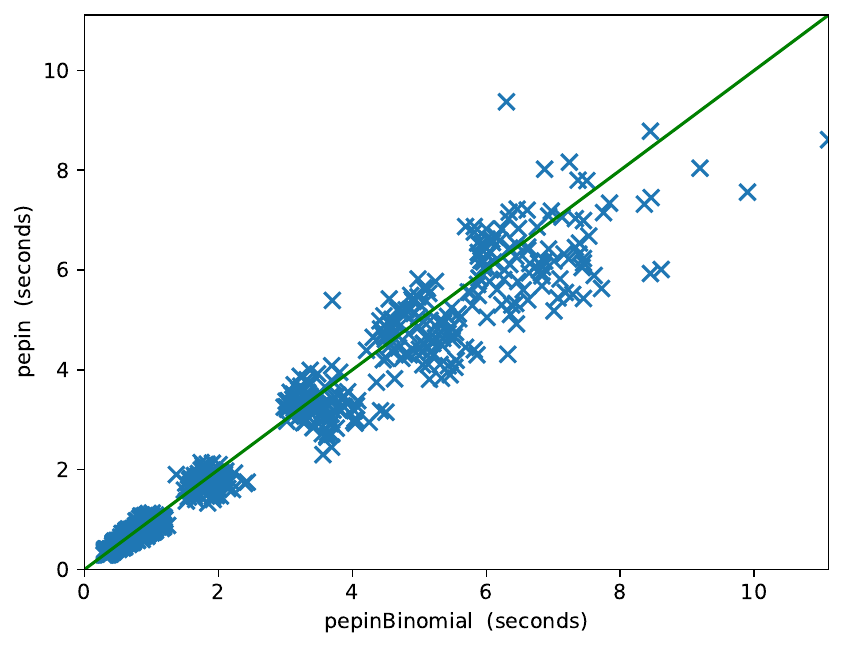}
		\caption{Instance-wise performance comparison}
        \label{fig:pepin-old-2}
    \end{subfigure}
	\caption{Performance comparison of {\toolname} against its earlier version {\oldtoolname}.}
	\label{fig:pepin-old}
\end{figure}

\begin{table}[h]
	\centering
	\begin{tabular}{p{0.13\textwidth}p{0.13\textwidth}}
 \toprule
Counter & PAR-2 score\\
 \midrule
{\toolname}         &2.16\\
{\oldtoolname}      &2.28\\
\bottomrule
\end{tabular}
\caption{Comparing the PAR2 scores of {\toolname} and {\oldtoolname} over the performance problems.}
\label{tbl:pepin-PAR2-scores}
\end{table}

Fig.~\ref{fig:pepin-old} illustrates the performance comparison between {\oldtoolname} and {\toolname} across all instances. 
Fig.~\ref{fig:pepin-old-1} illustrates the time taken by {\toolname} and {\oldtoolname} to solve all 900 instances. The X-axis denotes the number of solved instances and the Y-axis denotes the time consumed by the counters. The plot clearly depicts that {\oldtoolname} eventually requires more time to completely finish all the instances, indicating overall performance improvement of {\toolname} over {\oldtoolname}.

Fig.~\ref{fig:pepin-old-2} compares the instance-wise runtime performance of {\toolname} against {\oldtoolname}. The X-axis (resp. Y-axis) denotes the time consumed by {\oldtoolname} (resp. {\toolname}) to solve each instance. 
The plot shows that most of the points are below the X=Y line, which clearly indicates that for the majority of the instances, {\oldtoolname} requires more time than {\toolname}. 

Additionally, we analyzed the PAR2 score and observed that {\toolname} achieves a PAR2 score of 2.16 seconds, whereas {\oldtoolname} achieves a PAR2 score of 2.28 seconds, denoting around 5\% performance improvement of {\toolname} from {\oldtoolname}.

\newpage

\section{Poisson Tail Bounds}

\begin{lemma}
\label{lem:chernoff2}
Let $X \gets \pois(\lambda)$, for some $\lambda>0$. Then for any $x > 0$,
\[
\Pr[X \geq \lambda + x] \le e^{-\frac{x^2}{2\lambda + x}}\;
\]
and for any $0 < x < \lambda$,
\[
\Pr[X  \leq \lambda - x] \le e^{-\frac{x^2}{2\lambda + x}}\;
\]
\end{lemma}

\chernoff*

\begin{proof}
    By putting together the following inequalities: $\Pr[X \geq \lambda + x ] \leq e^{-\frac{x^2}{2\lambda + x}}$, $\Pr[X \geq \lambda - x ] \leq e^{-\frac{x^2}{2\lambda + x}}$ for $x >0$, we deduce that, 
    \[
    \Pr[|X - \lambda| \geq x] \le 2e^{-\frac{x^2}{2\lambda+x}}\;
    \]
\end{proof}

Before we proceed to prove \cref{lem:chernoff2}, we need to define the function $h: [-1, \infty) \to \mathbb{R}$ such that,  for $z \in [-1, \infty)$, $$h(z) = 2 \frac{(1+z)\log(1+z) - z}{z^2}$$ 
\begin{fact}[\cite{kozma}]
\label{fact:log}
For $x\geq 0$, the following inequality holds,
\[\log(1+x) \geq \frac{2x}{2+x}\]
\end{fact}

\begin{lemma}[\cite{canonne}]
\label{lem:chernoff3}
Let $X \gets \pois(\lambda)$, for some $\lambda>0$. Then for any $x > 0$,
\[
\Pr[X \geq \lambda + x] \le e^{-\frac{x^2}{2\lambda} h\left(\frac{x}{\lambda}\right)}\;
\]
and for any $0 < x < \lambda$,
\[
\Pr[X  \leq \lambda - x] \le e^{-\frac{x^2}{2\lambda} h\left(-\frac{x}{\lambda}\right)}\;
\]
\end{lemma}

\begin{proof}[Proof of \cref{lem:chernoff2}]
    To prove \cref{lem:chernoff2} we only need to prove that $h(z) \geq \frac{2}{2+z}$ for any $z \geq 0$. 
    This is equivalent to proving the following inequality:

    \begin{align*}
        &\frac{(1+z)\log(1+z) - z}{z^2} \cdot (2+z) \geq 1 \\
        \implies & \left((1+z)\log(1+z) - z\right)(2+z) - z^2 \geq 0 \\
        \implies & (1+z)(2+z)\log(1+z) - 2z(1+z) \geq 0 \\
        \implies & (2+z)\log(1+z) - 2z \geq 0 &&[\text{ followed by, $z \geq 0$}]\\
    \end{align*}
    The last inequality holds true based on \cref{fact:log}. Thus, the lemma is proven. 
\end{proof}


\begin{thebibliography}{10}
\expandafter\ifx\csname url\endcsname\relax
  \def\url#1{\texttt{#1}}\fi
\expandafter\ifx\csname urlprefix\endcsname\relax\def\urlprefix{URL }\fi
\expandafter\ifx\csname href\endcsname\relax
  \def\href#1#2{#2} \def\path#1{#1}\fi

\bibitem{DalviS07}
N.~N. Dalvi, D.~Suciu, Efficient query evaluation on probabilistic databases, {VLDB} J. 16~(4) (2007) 523--544.
\newblock \href {https://doi.org/10.1007/s00778-006-0004-3} {\path{doi:10.1007/s00778-006-0004-3}}.

\bibitem{Karger2001}
D.~R. Karger, \href{http://epubs.siam.org/doi/abs/10.1137/S0036144501387141}{{A Randomized Fully Polynomial Time Approximation Scheme for the All-Terminal Network Reliability Problem}}, SIAM Review 43~(3) (2001) 499--522.
\newblock \href {https://doi.org/10.1137/S0036144501387141} {\path{doi:10.1137/S0036144501387141}}.
\newline\urlprefix\url{http://epubs.siam.org/doi/abs/10.1137/S0036144501387141}

\bibitem{Valiant79}
L.~G. Valiant, The complexity of enumeration and reliability problems, SIAM Journal on Computing 8~(3) (1979) 410--421.

\bibitem{KarpLuby1983}
R.~Karp, M.~Luby, Monte-carlo algorithms for enumeration and reliability problems, Proc. of {FOCS} (1983).

\bibitem{KLM89}
R.~M. Karp, M.~Luby, N.~Madras, Monte-carlo approximation algorithms for enumeration problems, Journal of Algorithms 10~(3) (1989) 429 -- 448.

\bibitem{booksVV}
V.~V. Vazirani, \href{http://www.springer.com/computer/theoretical+computer+science/book/978-3-540-65367-7}{Approximation algorithms}, Springer, 2001.
\newline\urlprefix\url{http://www.springer.com/computer/theoretical+computer+science/book/978-3-540-65367-7}

\bibitem{CMV16}
S.~Chakraborty, K.~S. Meel, M.~Y. Vardi, Algorithmic improvements in approximate counting for probabilistic inference: From linear to logarithmic {SAT} calls, in: Proc. of IJCAI, 2016.

\bibitem{MSV17}
K.~S. Meel, A.~A. Shrotri, M.~Y. Vardi, On hashing-based approaches to approximate {DNF}-counting, in: Proc. of FSTTCS, Vol.~93, 2017, pp. 41:1--41:14.

\bibitem{MSV19}
K.~S. Meel, A.~A. Shrotri, M.~Y. Vardi, Not all {FPRASs} are equal: demystifying {FPRASs} for {DNF}-counting, Constraints An Int. J. 24~(3-4) (2019) 211--233.

\bibitem{MVC21}
K.~S. Meel, N.~V. Vinodchandran, S.~Chakraborty, Estimating the size of union of sets in streaming models, in: PODS, 2021, pp. 126--137.
\newblock \href {https://doi.org/10.1145/3452021.3458333} {\path{doi:10.1145/3452021.3458333}}.

\bibitem{S83}
L.~J. Stockmeyer, The complexity of approximate counting (preliminary version), in: STOC, 1983, pp. 118--126.
\newblock \href {https://doi.org/10.1145/800061.808740} {\path{doi:10.1145/800061.808740}}.

\bibitem{Sip83}
M.~Sipser, A complexity theoretic approach to randomness, in: STOC, 1983, pp. 330--335.
\newblock \href {https://doi.org/10.1145/800061.808762} {\path{doi:10.1145/800061.808762}}.

\bibitem{OHK10}
D.~Olteanu, J.~Huang, C.~Koch, Approximate confidence computation in probabilistic databases, in: ICDE, 2010, pp. 145--156.
\newblock \href {https://doi.org/10.1109/ICDE.2010.5447826} {\path{doi:10.1109/ICDE.2010.5447826}}.

\bibitem{FO11}
R.~Fink, D.~Olteanu, On the optimal approximation of queries using tractable propositional languages, in: ICDT, 2011, pp. 174--185.
\newblock \href {https://doi.org/10.1145/1938551.1938575} {\path{doi:10.1145/1938551.1938575}}.

\bibitem{GS14}
W.~Gatterbauer, D.~Suciu, Oblivious bounds on the probability of boolean functions, {ACM} Trans. Database Syst. 39~(1) (2014) 5:1--5:34.
\newblock \href {https://doi.org/10.1145/2532641} {\path{doi:10.1145/2532641}}.

\bibitem{TSVO04}
Q.~Tao, S.~D. Scott, N.~V. Vinodchandran, T.~T. Osugi, Svm-based generalized multiple-instance learning via approximate box counting, in: ICML, Vol.~69 of {ACM} International Conference Proceeding Series, {ACM}, 2004.
\newblock \href {https://doi.org/10.1145/1015330.1015405} {\path{doi:10.1145/1015330.1015405}}.

\bibitem{LubyV96}
M.~Luby, B.~Velickovic, On deterministic approximation of {DNF}, Algorithmica 16~(4/5) (1996) 415--433.
\newblock \href {https://doi.org/10.1007/BF01940873} {\path{doi:10.1007/BF01940873}}.

\bibitem{Trevisan04}
L.~Trevisan, A note on approximate counting for {k-DNF}, in: APPROX, Vol. 3122, 2004, pp. 417--426.
\newblock \href {https://doi.org/10.1007/978-3-540-27821-4\_37} {\path{doi:10.1007/978-3-540-27821-4\_37}}.

\bibitem{GopalanMR13}
P.~Gopalan, R.~Meka, O.~Reingold, {DNF} sparsification and a faster deterministic counting algorithm, Comput. Complex. 22~(2) (2013) 275--310.
\newblock \href {https://doi.org/10.1007/s00037-013-0068-6} {\path{doi:10.1007/s00037-013-0068-6}}.

\bibitem{GNUBignum}
T.~Granlund, the GMP~development team, The {GNU} multiple precision arithmetic library, Website, \url{https://gmplib.org/gmp-man-6.2.1.pdf} (Nov 2020).

\bibitem{DKLR00}
P.~Dagum, R.~Karp, M.~Luby, S.~Ross, An optimal algorithm for monte carlo estimation, SIAM Journal on computing 29~(5) (2000) 1484--1496.

\bibitem{SACMO23}
M.~Soos, D.~Aggarwal, S.~Chakraborty, K.~S. Meel, M.~Obremski, Engineering an efficient approximate dnf-counter, in: IJCAI, 2023, pp. 2031--2038.

\bibitem{kozma}
L.~Kozma, \href{https://www.lkozma.net/inequalities_cheat_sheet/ineq.pdf}{{Useful inequalities cheat sheet}} (2023).
\newline\urlprefix\url{https://www.lkozma.net/inequalities_cheat_sheet/ineq.pdf}

\bibitem{canonne}
C.~Canonne, \href{https://github.com/ccanonne/probabilitydistributiontoolbox/blob/master/poissonconcentration.pdf}{{A short note on Poisson tail bounds}} (2019).
\newline\urlprefix\url{https://github.com/ccanonne/probabilitydistributiontoolbox/blob/master/poissonconcentration.pdf}

\end{thebibliography}
\end{document}